\documentclass[11pt]{article}
\usepackage{fullpage}

\usepackage{hyperref,url}       
\usepackage{url}            
\usepackage{booktabs}       
\usepackage{nicefrac}       
\usepackage{microtype}      

\usepackage{microtype}
\usepackage{algorithm,algorithmic}
\usepackage{amsmath,amssymb,amsthm,amsfonts}
\usepackage{xspace}
\usepackage{graphicx,subfigure}
\usepackage{mdframed}
\definecolor{light-gray}{gray}{0.95}
\usepackage{mathtools}
\usepackage{tabu,multirow}
\usepackage{dblfloatfix}
 \usepackage{epsfig}
\usepackage{xcolor,color,url,eurosym,bbm}
\usepackage{array,multirow,tikz,wrapfig}
\usetikzlibrary{shapes.geometric}
\usepackage{empheq}
\usepackage{adjustbox}
\usepackage[graphicx]{realboxes}

\usepackage[most]{tcolorbox}
\newtcbox{\mymath}[1][]{%
    nobeforeafter, math upper, tcbox raise base,
    enhanced, colframe=blue!30!black,
    colback=blue!30, boxrule=1pt,
    #1}

 \usepackage{thm-restate}

\newtheorem{theorem}{Theorem}[section]
\newtheorem{lemma}[theorem]{Lemma}
\newtheorem{corollary}[theorem]{Corollary}

\theoremstyle{definition}

 \usepackage{color}

\makeatletter
\newcommand*{\rom}[1]{\expandafter\@slowromancap\romannumeral #1@}
\makeatother

\newcommand\myeq{\mathrel{\stackrel{\makebox[0pt]{\mbox{\normalfont\tiny def}}}{=}}}

\newcommand{\squishlist}{
 \begin{list}{$\bullet$}
  {  \setlength{\itemsep}{0pt}
     \setlength{\parsep}{3pt}
     \setlength{\topsep}{3pt}
     \setlength{\partopsep}{0pt}
     \setlength{\leftmargin}{2em}
     \setlength{\labelwidth}{1.5em}
     \setlength{\labelsep}{0.5em}
} }
\newcommand{\squishlisttight}{
 \begin{list}{$\bullet$}
  { \setlength{\itemsep}{0pt}
    \setlength{\parsep}{0pt}
    \setlength{\topsep}{0pt}
    \setlength{\partopsep}{0pt}
    \setlength{\leftmargin}{2em}
    \setlength{\labelwidth}{1.5em}
    \setlength{\labelsep}{0.5em}
} }

\newcommand{\squishdesc}{
 \begin{list}{}
  {  \setlength{\itemsep}{0pt}
     \setlength{\parsep}{3pt}
     \setlength{\topsep}{3pt}
     \setlength{\partopsep}{0pt}
     \setlength{\leftmargin}{1em}
     \setlength{\labelwidth}{1.5em}
     \setlength{\labelsep}{0.5em}
} }

\newcommand{\squishend}{
  \end{list}
}

\newcommand{\squishlistt}{
 \begin{list}{---}
  {  \setlength{\itemsep}{0pt}
     \setlength{\parsep}{3pt}
     \setlength{\topsep}{3pt}
     \setlength{\partopsep}{0pt}
     \setlength{\leftmargin}{2em}
     \setlength{\labelwidth}{1.5em}
     \setlength{\labelsep}{0.5em}
} }
\newcommand{\randG}{\textbf{g}}
\newcommand{\Alg}{\mathcal{A}}	
\newcommand{\eps}{\varepsilon}
\newcommand{\randF}{\textbf{f}}
\DeclareMathOperator*{\supp}{supp}

\newcommand{\hide}[1]{} 

\newcommand{\beql}[1]{\begin{equation}\label{#1}}

\newcommand{\beq}[1]{\begin{equation}\label{#1}}
\newcommand{\eeq}{\end{equation}}

\newcommand{\Prob}[1]{\ensuremath{{\bf{Pr}}\left[{#1}\right]}}
\newcommand{\Mean}[1]{\ensuremath{{\mathbb E}\left[{#1}\right]}}

\newcommand{\whp}{\textit{whp}\xspace}

\newcommand{\spara}[1]{\smallskip\noindent{\bf #1}}

\begin{document}

\title{Optimal Learning of Joint Alignments with a Faulty Oracle}

\author{Kasper Green Larsen\thanks{Aarhus University,  {\texttt{larsen@cs.au.dk}}.}        
\and    
Michael Mitzenmacher
\thanks{Harvard  University,  {\texttt{michaelm@eecs.harvard.edu}}.}
\and
Charalampos E. Tsourakakis\thanks{Boston University,  {\texttt{ctsourak@bu.edu}}.}
}

\maketitle

\begin{abstract}
We consider the following problem, which is useful in applications such as joint image and shape alignment.  The goal is to recover $n$
discrete variables $g_i \in \{0, \ldots, k-1\}$ (up to some global offset) given noisy observations of a set of their pairwise
differences $\{(g_i - g_j) \bmod k\}$; specifically, with probability $\frac{1}{k}+\delta$ for some $\delta > 0$ one obtains the correct answer, and with the remaining probability one obtains a uniformly random incorrect answer.  We consider a learning-based formulation where one can perform a query to observe a pairwise difference, and the goal is to perform as few queries as possible while obtaining the exact joint alignment.  We provide an easy-to-implement, time efficient algorithm that performs $O\big(\frac{n \lg n}{k \delta^2}\big)$ queries, and recovers the joint alignment with high probability.  We also show that our algorithm is
optimal by proving a general lower bound that holds for all non-adaptive algorithms. Our work improves significantly recent work
by Chen and Cand\'{e}s \cite{chen2016projected}, who view the problem as a constrained principal components analysis problem that can be solved using the power method. Specifically, our approach is simpler both in the algorithm and the analysis, and provides additional insights into the problem structure. 
\end{abstract}

\pagenumbering{gobble}

\vfill

\pagebreak

\pagenumbering{arabic}
 
\section{Introduction}  
\label{sec:intro} 
Learning a joint alignment from pairwise differences is a problem with various important applications  ranging from shape matching \cite{huang2013fine}, to spectroscopy imaging \cite{wang2013exact}. In this work we adopt the following established mathematical formalization of this problem. There exists a set $V=[n]$ of $n$ discrete items, and an assignment $g:V \rightarrow [k]$ according to which each  item   is assigned one out of $k$ possible labels.   To give an example, imagine a set of $n$ images of the same object in $k$ possible orientations/angles, where each $g(i)$ is one of $k$ possible orientations (angles) of the camera  when taking the $i$-th image. Recovering $g$ would allow a better understanding of the 3-dimensional structure of the object.  The assignment function $g$ is unknown, but we may obtain a set of pairwise noisy difference samples $\{\tilde{f}(i,j) \myeq  (g(i)-g(j)+ \text{noise}) \bmod k\}_{(i,j) \in \Omega}$ where $\Omega \subseteq {[n] \choose 2}$ is a symmetric index set, i.e., a set of pairs $\{i,j\}$ with $i < j$.  In this work, we consider the setting where each pair can be queried at most once (e.g., the measurement will not change on repeated queries), and the noisy measurement $\tilde{f}(x,y) $ is equal to 
\begin{equation} 
\label{eq:model}
    \tilde{f}(x,y) =  
        \big(g(x)-g(y)+ \eta_{xy}\big) \bmod k
\end{equation} 
\noindent where the additive noise values $\eta_{xy} $ are
i.i.d.~random variables supported on $\{0,1,\cdots,k-1\}$, with
the following probability distribution that is slightly biased
towards zero for some parameter $\delta > 0$:
\begin{equation} 
  \label{eq:noise}
  \Prob{\eta_{xy}=i} = \left\{\begin{array}{lr}
        \frac{1}{k}+\delta, & \text{if }i=0;\\
        \frac{1}{k}-\frac{\delta}{k-1}, & \text{for each }i\neq 0.\\
         \end{array}\right.
\end{equation} 

In this work we study the problem of recovering $g$ up to some global offset by choosing the set of queries $\Omega$.

\spara{Related Work.}  Learning joint alignments is a major problem that appears in numerous settings under different guises. In cryo-electron microscopy, the problem corresponds to recovering the angles from which 2d pictures of a 3d object were taken. This allows for the construction of a 3d model of the objective \cite{shkolnisky2012viewing}.  In shape matching, a key problem is assembling fractured surfaces \cite{huang2006reassembling} and fusing scans to model reality \cite{huber2002automatic}, jointly optimizing the maps between shapes improves the performance compared to matching shapes in isolation \cite{huang2013consistent}.

Closest to our work lies the work of Chen and Cand\`{e}s \cite{chen2016projected}, who study the same model (Equation~\eqref{eq:model}\footnote{The parameter $\pi_0$ in their random corruption model, and our bias $\delta$ are connected with the following equation $\delta= \pi_0 \frac{k-1}{k}$.}).  They provide an algorithm that is non-adaptive, and the underlying queries form a random binomial graph, i.e. each edge as queried independently with a fixed probability. They show that, in the setting where queries form a random binomial graph, the minimax probability of error tends to 1 if the number of queries is less than $\Omega\left (\frac{n \log n}{k \delta^2} \right )$ \cite[Theorem 2,p. 7]{chen2016projected}. Their algorithm, based on the projected power method, has a required number of queries that matches the lower bound.  Inferior results have been obtained in the past as well. Notably, a simpler non-adaptive algorithm with somewhat inferior query complexity that relies on simple breadth-first search was proposed by Mitzenmacher and Tsourakakis \cite{mitzenmacher2018joint}. Chen et al. provide an SDP-based algorithm  \cite{chen2014near} that is slower and with more stringent recovery conditions than \cite{chen2016projected}.   A closely related but different approach with respect to the mathematical formulation is the  phase/angular synchronization problem \cite{singer2011angular,zhong2018near}. It is worth remarking that the special case $k=2$ reduces to an active learning problem related to graph partitioning problem that is well-studied, e.g. \cite{mazumdar2017clustering,tsourakakis2017predicting}, with close connections to the classic planted partition problem \cite{abbe2015community,hajek2016achieving,mcsherry2001spectral,tsourakakis2015streaming}.  

\spara{Our Results.}  In this paper we provide a simpler non-adaptive algorithm that we prove also succeeds with high probability with $O\left (\frac{n \log n}{k \delta^2} \right )$ queries.  Our algorithm is based on selecting a small seed set and using queries to obtain and reconcile all edge measurements for edges adjacent to these vertices;  this approach itself appears of interest.  We also provide a simpler lower bound argument showing our result is tight in terms of the number of queries required in this more general setting where queries are arbitrary.

\section{Proposed Method} 
\label{sec:prop} 
\subsection{Preliminaries}
\label{sec:prelim}

Both our algorithm and our lower bound proof need tight concentration inequalities on the probability that the majority of a collection of biases $\eta_{xy}$ is equal to $0$. We state the two concentration inequalities here.
\hide{
 The proofs are provided in the supplementary material. 
}
The proofs are in Section~\ref{sec:concentration}.
The first lemma considers the case of small $\delta$:
\begin{lemma}
\label{lem:smalldelta}
Let $k \geq 2$ be an integer, let $0 \leq \delta \leq 1/2k$ and let $X_1,\dots,X_n$ be i.i.d. random variables such that each $X_i$ takes the value $1$ with probability $1/k + \delta$, the value $-1$ with probability $1/k-\delta/(k-1)$ and the value $0$ otherwise. There exists constants $c_1,c_2> 0$ such that:
$$
\Pr[\sum_i X_i \leq 0] \leq c_1 \exp(-\delta^2 nk/c_1)
$$
and
$$
\Pr[\sum_i X_i \leq 0] \geq c_2^{-1} \exp(-\delta^2 nk c_2).
$$
\end{lemma}
And the second considers the case of large $\delta$:
\begin{lemma}
\label{lem:bigdelta}
Let $k \geq 2$ be an integer, let $1/2k < \delta \leq  1/4$ and let $X_1,\dots,X_n$ be i.i.d. random variables such that each $X_i$ takes the value $1$ with probability $1/k + \delta$, the value $-1$ with probability $1/k-\delta/(k-1)$ and the value $0$ otherwise. There exists constants $c_1,c_2 > 0$ such that:
$$
\Pr[\sum_i X_i \leq 0] \leq c_1 \exp(-\delta n/c_1)
$$
and
$$
\Pr[\sum_i X_i \leq 0] \geq c_2^{-1}\exp(-\delta n c_2).
$$
\end{lemma}

\subsection{Upper bound - Proposed Algorithm} 
Our algorithm is a simple and efficient non-adaptive algorithm. The basic idea is to choose a set of nodes $S$ as a \emph{seed} set of nodes. We then make all queries between $S$ and the full node set $V$. Based on these queries, we first determine the label $g(s)$ of all nodes $s \in S$ (up to a cyclic shift). Once these have been determined, we can determine the labels of all remaining nodes $v$ by using a plurality vote on $\{g(s) + \tilde{f}(v,s) \bmod k \mid s \in S\}$. We proceed to give the details.

\begin{lemma}[Plurality vote]
\label{lem1} 
Let $S \subseteq V$ be an arbitrary seed set of nodes and assume $k \leq n^{o(1)}$. For any node $v \in V \setminus S$, the plurality vote among $\{ g(s) + \tilde{f}(v,s) \bmod k \mid s \in S\}$ is equal to $g(v)$ with probability at least $1- \frac{1}{n^2}$ if either:
\begin{itemize}
\item $0 \leq \delta \leq 1/2k$ and $|S| = \Omega(\frac{\lg n}{\delta^2 k})$, or
\item $1/2k \leq \delta \leq 1/4$ and $|S| = \Omega(\frac{\lg n}{\delta})$.
\end{itemize}
\end{lemma}

By taking a union bound over all nodes $v \notin S$, we obtain the following straight-forward corollary. 

\begin{corollary} 
\label{cor:seed}
Assume we have a seed of nodes $S$, such that for all $s \in S$, we know $g(s) + \alpha \bmod k$ for some (shared) cyclic shift $\alpha \in \{0,\dots,k-1\}$. Then it is possible to recover $g(v) + \alpha \bmod k$ for all $v \in V $ in $O(n|S|)$ time \whp. provided that $k \leq n^{o(1)}$ and either:
\begin{itemize}
\item $0 \leq \delta \leq 1/2k$ and $|S| = \Omega(\frac{\lg n}{\delta^2 k})$, or
\item $1/2k \leq \delta \leq 1/4$ and $|S| = \Omega(\frac{\lg n}{\delta})$.
\end{itemize}
\end{corollary}

\begin{proof}[Proof of Lemma~\ref{lem1}]
Each query $\tilde{f}(v,s)$ returns $(g(v) - g(s) + \eta_{vs}) \bmod k$. We thus have $(g(s) + \tilde{f}(v,s)) \bmod k = (g(v) + \eta_{vs}) \bmod k$. Therefore $(g(s) + \tilde{f}(v,s)) \bmod k = g(v)$ with probability $1/k + \delta$, and for every $i \in \{1,\dots,k-1\}$, we have $(g(s) + \tilde{f}(v,s)) \bmod k = (g(v) + i) \bmod k$ with probability $1/k - \delta/(k-1)$. Using Lemma~\ref{lem:smalldelta} and a union bound over all $i \in \{1,\dots,k-1\}$, we thus conclude for $0 \leq \delta \leq 1/2k$, that the plurality vote equals $g(v)$ with probability at least $1-k c_1 \exp(-\delta^2 |S|k/c_1)$ for a constant $c_1 > 0$. For $|S|= \Omega(\frac{\lg n}{\delta^2 k})$ and $k \leq n^{o(1)}$, this is at least $1-1/n^2$. Similarly we use Lemma~\ref{lem:bigdelta} and a union bound over all $i \in \{1,\dots,k-1\}$ to conclude for $1/2k \leq \delta \leq 1/4$, that the plurality vote equals $g(v)$ with probability at least $1-k c_1 \exp(-\delta |S|/c_1)$ for a constant $c_1 > 0$. For $|S| = \Omega(\frac{\lg n}{\delta})$ and $k \leq n^{o(1)}$, this is at least $1-1/n^2$.
\end{proof}

\noindent Given Corollary~\ref{cor:seed} it suffices to find a seed set $S \subseteq V$ and determine the labels of the nodes in $S$ up to the same cyclic shift $\alpha$. Our next lemma shows how to do so via queries $\tilde{f}(s,v)$ for nodes $s \in S$ and $v \in V \setminus S$. Our key idea is to determine the difference $(g(s) - g(s')) \bmod k$ for pairs $s,s' \in S$ via queries $\tilde{f}(s,b)-\tilde{f}(s',b)$ for nodes $b \in V \setminus S$.

\begin{lemma}[Learning Pairwise Differences]
\label{lem2}
Let $S \subseteq V$ be an arbitrary set of nodes and assume $k \leq n^{o(1)}$. Let $s,s' \in S$ be two distinct nodes. Define  $Z_{s,s'}$ as the plurality vote among the answers $\{ (\tilde{f}(s,b)-\tilde{f}(s',b)) \bmod k \}_{b \in V \setminus S}$. If $|V \setminus S| = \Omega(\frac{\lg n}{k \delta^4} + \frac{\lg n}{\delta^2})$, then $Z_{a,a'}=(g(a)-g(a')) \bmod k$ with probability at least $1-\frac{1}{n^2}$.
\end{lemma}

\begin{proof}
Since
\begin{align*}
\tilde{f}(s,b)-\tilde{f}(s',b) &= (g(s)-g(s') \bmod k)   +(\eta_{s,b}-\eta_{s',b}\bmod k),
\end{align*}
we need to understand the probability distribution of $Z_b=\eta_{s,b}-\eta_{s',b}\bmod k$. Intuitively, we wish that the probability $\Pr[Z_b=0]$ is greater enough than each $\Pr[Z_b=i]$ where $ i \neq 0$  so that the plurality vote gives the correct estimate for $g(s)-g(s')$. Indeed, 

\begin{align*}
\Pr[Z_b = 0] &= \sum_{j=0}^{k-1} \Pr[\eta_{s,b}=\eta_{s',b}=j] = \Pr[\eta_{s,b}=\eta_{s',b}=0] + \sum_{j=1}^{k-1} \Pr[\eta_{s,b}=\eta_{s',b}=j] = \\ 
&=\Big(\frac{1}{k}+\delta \Big)^2 + (k-1) \Big( \frac{1}{k}-\frac{\delta}{k-1}\Big)^2 = \frac{1}{k}+\frac{k\delta^2}{k-1}.
\end{align*}

Also $Z_b$ is uniform over $1,\dots,k-1$ with the remaining probability, i.e. $\Pr[Z_b = i] = \frac{1}{k}-\frac{k \delta^2}{(k-1)^2}$ for $i \neq 0$. We thus obtain the exact same guarantees as in Lemma~\ref{lem1} with $\delta$ replaced by $\delta' = \frac{k \delta^2}{k-1}$. That is, if either
\begin{itemize}
\item $0 \leq \frac{k \delta^2}{k-1} \leq 1/2k$ and $|V \setminus S| = \Omega(\frac{(k-1)^2 \lg n}{ k^3 \delta^4})$, or
\item $1/2k \leq \frac{k \delta^2}{k-1} \leq 1/4$ and $|V \setminus S| = \Omega(\frac{(k-1) \lg n}{ k \delta^2})$.
\end{itemize}
then the plurality vote among $\{ (\tilde{f}(s,b)-\tilde{f}(s',b)) \bmod k \}_{b \in B}$ equals $(g(s)-g(s')) \bmod k$ with probability at least $1-1/n^2$. Combining the two, we conclude from the above that the plurality vote is correct with probability at least $1-1/n^2$ provided that $|V \setminus S| = \Omega(\frac{(k-1)^2 \lg n}{ k^3 \delta^4} + \frac{(k-1) \lg n}{ k \delta^2}) = \Omega(\frac{\lg n}{k \delta^4} + \frac{\lg n}{\delta^2})$.
\end{proof}

In light of the above, our proposed algorithm is thus to pick a set $S$ and perform all queries between $S$ and $V \setminus S$. Based on Lemma~\ref{lem1} and Lemma~\ref{lem2}, we set $|S| = O(\frac{\lg n}{k \delta^2})$ when $0 \leq \delta \leq 1/2k$ and $|S| = O(\frac{\lg n}{\delta})$ when $1/2k \leq \delta \leq 1/4$. We then fix a node $s \in S$ and assign it the label $\hat{g}(s) = 0$. We thus have $\hat{g}(s) = (g(s) + (0 - g(s))) \bmod k$, i.e. $g(s)$ has been recovered up to a cyclic shift of $(0-g(s))$. Our goal is to recover all other labels up to the same cyclic shift.

We now compute an estimate $\mu_{s'}$ of $(g(s)-g(s')) \bmod k$ for every $s' \in S \setminus \{s\}$ using a plurality vote on $\{ (\tilde{f}(s,b)-\tilde{f}(s',b)) \bmod k \}_{b \in V \setminus S}$. A union bound over all nodes in $S$ together with Lemma~\ref{lem2} shows that all these estimates are correct \whp. We then assign the label $\hat{g}(s') = \mu_{s'}$ to all remaining nodes $s' \in S$. If all plurality votes were correct, then $\hat{g}(s') = \mu_{s'} = (g(s')-g(s)) \bmod k = (g(s') + (0 - g(s)))\bmod k$ for all $s'$. That is, we have recovered each $g(s')$ up to the same cyclic shift $(0-g(s)) \bmod k$.

To recover the labels of all remaining nodes $v \in V \setminus S$ in the graph (up to the shift $(0-g(s)) \bmod k$), we use a plurality vote on $\{\hat{g}(s') + \tilde{f}(v,s') \bmod k\}_{s' \in S} = \{g(s') + (0-g(s)) + \tilde{f}(v,s') \bmod k\}_{s' \in S}$. Corollary~\ref{cor:seed} and a union bound over all nodes in $V \setminus S$ gives us that our algorithm recovers all labels \whp. Our proposed algorithm is also shown in pseudocode, see Algorithm~\ref{alg:3cc}. 

\begin{algorithm}[h]
\caption{\label{alg:3cc} Learning Joint Alignment with a Faulty Oracle} 
 \begin{algorithmic} 
 \STATE Choose $S \subseteq V$ such that $|S|=O(\frac{\log n}{k \delta^2})$ if $0 \leq \delta \leq 1/2k$ and
 $|S| =  O(\frac{\lg n}{\delta})$ if $1/2k \leq \delta \leq 1/4$. 
 \STATE Perform all queries between $S$ and $V \setminus S$.
 \STATE Fix a node $s \in S$ and assign it the label $\hat{g}(s)=0$.
\STATE For each $s' \in S \setminus \{s\}$, compute an estimate $\mu_{s'}$ of $(g(s')-g(s)) \bmod k$ using the plurality vote among the queries $\{ \tilde{f}(s',b)-\tilde{f}(s,b) \}_{b \in V \setminus S}$ and assign $s'$ the label $\hat{g}(s')=\mu_{s'}$.
 \STATE For each $v  \notin V \setminus S$, assign it a label corresponding to the result of the plurality vote among $\{ \hat{g}(s) + \tilde{f}(v,s)\}_{s \in S}$. 
\end{algorithmic}
\end{algorithm}

As a last remark, notice that we can only choose $|S| = O(\frac{\lg n}{k \delta^2})$ or $|S| = O(\frac{\lg n}{\delta})$ provided that $\frac{\lg n}{k \delta^2} = O(n)$ in the first case and $\frac{\lg n}{\delta} = O(n)$ in the second case. Assume first that indeed $|S| \leq n/2$. Then Lemma~\ref{lem2} further requires that $|V \setminus S| = \Omega(\frac{\lg n}{k \delta^4} + \frac{\lg n}{\delta^2})$. Since $|V \setminus S| \geq n/2$ when $|S| \leq n/2$, this translates into $\frac{\lg n}{k \delta^4} + \frac{\lg n}{\delta^2} = O(n)$. This is a more strict requirement than $\frac{\lg n}{k \delta^2} = O(n)$ and $\frac{\lg n}{\delta} = O(n)$. We can thus invoke our algorithm as long as $\delta = \Omega((\lg n/nk)^{1/4})$ and $\delta = \Omega(\sqrt{1/n})$. We assume $k \leq n^{o(1)}$, hence the dominating requirement is $\delta = \Omega((\lg n/nk)^{1/4})$.

The algorithm is completely non-adaptive, correct \whp. and each plurality vote can be computed in linear time in the number of estimates involved. The total running time of the algorithm is thus $O(|V||S|)$ and so is the number of queries. When $0 \leq \delta \leq 1/2k$, this is $O(\frac{n \lg n}{\delta})$ and when $1/2k \leq \delta \leq 1/4$, this is $O(\frac{n \lg n}{k \delta^2})$.

\begin{theorem}
If $(\lg n/nk)^{1/4} \leq \delta \leq 1/2k$ and $k \leq n^{o(1)}$, then there is a non-adaptive and deterministic query algorithm that makes $O(\frac{n\log n}{\delta^2 k})$ queries, runs in $O(\frac{n\log n}{\delta^2 k})$ time and is correct \whp.

If $1/2k \leq \delta \leq 1/4$ and $k \leq n^{o(1)}$, then there is a non-adaptive and deterministic query algorithm that makes $O(\frac{n \log n}{\delta})$ queries, runs in $O(\frac{n\log n}{\delta})$ time and is correct \whp.
\end{theorem}

\subsection{Lower bound } 
In this section, we complement our algorithm with a matching lower bound:
\begin{theorem}
If $1/n^{1/4} \leq \delta \leq 1/2k$ and $k \leq n^{o(1)}$, then any non-adaptive and possibly randomized query algorithm making $o(\frac{n\log n}{\delta^2 k})$ queries has success probability at most $\exp(-n^{\Omega(1)})$.

If $1/2k \leq \delta \leq 1/4$ and $k \leq n^{o(1)}$, then any non-adaptive and possibly randomized query algorithm making $o(\frac{n \log n}{\delta})$ queries has success probability at most $\exp(-n^{\Omega(1)})$.
\end{theorem}
Let $n$ be the number of vertices and consider a (possibly randomized) non-adaptive query algorithm $\Alg$, i.e. an algorithm that chooses the set of queries to make before seeing the results of the queries. Let $\eps$ be the success probability of $\Alg$, that is, for any latent function $g$, $\Alg$ recovers $g$ (up to a cyclic rotation of the labels) with probability at least $\eps$. Let $t$ be the number of queries made by $\Alg$. The choice of queries is allowed to be randomized. Our goal is to show that $\eps$ is small if $t$ is small.

\textbf{Hard Distribution.}
We start by defining a hard distribution. Let $\randG$ be a random latent function that assigns label $0$ to the first vertex and a uniform random and independently chosen label in $\{0,\dots,k-1\}$ to the remaining vertices. 

\textbf{Simplifying $\Alg$.}
Our first step is to simplify $\Alg$ for a cleaner analysis. Recall that a correct algorithm is allowed to return any cyclic rotation of the latent function $g$, i.e. any labeling that is equal to $g$ up to adding the same constant mod $k$ to all labels. Under our hard distribution $\randG$, we always have that the first vertex has label $0$. Therefore, we can define a new algorithm $\Alg^1$ which makes the same queries as $\Alg$, but when returning an assignment of labels, $\Alg^1$ takes the output of $\Alg$ and subtracts the label assigned by $\Alg$ to the first vertex from every single output label, mod $k$. In this way, for every $g \in \supp(\randG)$, we get that $\Alg^1$ returns $g$ whenever $\Alg$ is correct up to a cyclic rotation. That is, we now have an algorithm $\Alg^1$ that makes $t$ queries and has success probability $\eps$ for any $g \in \supp(\randG)$, even if we define success as returning the exact labeling (i.e. no cyclic shifts allowed). Our next simplifying step is to derandomize $\Alg^1$. By fixing the random coins of $\Alg^1$ (easy direction of Yao's principle), we obtain a deterministic algorithm $\Alg^2$ that makes $t$ non-adaptive queries and is correct with probability $\eps$ \emph{over the random choice of} $\randG$. Since $\Alg^2$ is deterministic and non-adaptive, we can let $E$ be the set of edges queried by $\Alg^2$ and let $\randF \in E \to \{0,\dots,k-1\}$ give the (random) results of the queries $E$. 

We wish to simplify $\Alg^2$ even further by making assumptions about the labeling it returns when seeing a set of answers $f \in \supp(\randF)$ to queries. Let $S$ denote the event that $\Alg^2$ is correct. Then 
\begin{eqnarray*}
\Pr[S] = \sum_{f \in \supp(\randF)} \Pr[\randF = f] \Pr[S \mid \randF=f].
\end{eqnarray*}
Since $\Alg^2$ is deterministic, it outputs a concrete labeling $\Alg^2(f)$ for any $f \in \randF$. Thus
\begin{eqnarray*}
\Pr[S \mid \randF =f] = \Pr[\randG = \Alg^2(f) \mid \randF=f].
\end{eqnarray*}
Let $G(f)$ be the collection of all maximum likelihood labelings $g \in \randG$ conditioned on $\randF =f$, i.e. $G(f)$ contains all $g \in \supp(\randG)$ such that $\Pr[\randG =g \mid \randF=f] \geq \Pr[\randG=g' \mid \randF=f]$ for all $g' \in \supp(\randG)$. The above allows us to conclude that if we define the algorithm $\Alg^*$ which makes the same queries as $\Alg^2$, but always returns a uniform random $g \in G(\randF)$, then $\Alg^*$'s success probability is at least $\eps$. This completes our simplifying steps and we will show that $\Alg^*$ has small success probability if $t$ is small.

\textbf{Performance of $\Alg^*$.}
To prove that $\Alg^*$ has low success probability if it makes few queries, we will show that there is a good chance that the correct labeling $\randG$ is not the maximum likelihood estimate after seeing $\randF$. For this, consider a vertex $v$ different from the first vertex and let $E_v$ be the subset of edges from $E$ that have $v$ as an end point. Since each edge in $E$ has two end points, there must be a set $W$ of at least $n/2$ vertices that have $|E_v| \leq 4t/n$. We form an independent set $I$ from $W$ by repeatedly selecting one vertex $v$ from $W$ and adding it to an initially empty $I$. We then remove all vertices incident to $v$ from $W$. Since each $v$ removes at most $4t/n$ other vertices from $W$, we are left with an $I$ of size at least $(n-1)/(4t/n+1)$. The reason why we choose $I$ as an independent set is that it implies that the queries corresponding to edges incident to a node $v \in I$ are independent of the queries incident to any other node $w \in I$.

Now let $f \in \supp(\randF)$ be an assignment to the edges and let $g \in \supp(\randG)$ be a classification of the vertices. Define from $f$ and $g$ the noise on edge $(u,v) \in E_v$ as $\eta^{fg}_{uv}=(g(u)-g(v)-f(u,v)) \bmod k$. For each $i \in \{0,\dots,k-1\}$, define $c^{fg}_v(i)$ as the number of edges $(u,v)$ incident to $v$ for which $\eta^{fg}_{uv} = i$. Define the subset $I^*_{fg} \subseteq I$ containing all vertices $v \in I$ such that $c^{fg}_v(1) \geq c^{fg}_v(0)$. We claim that there are at least $2^{|I^*_{fg}|}$ distinct labelings $g' \in \supp(\randG)$ that all have $\Pr[\randG = g' \mid \randF = f] \geq \Pr[\randG=g \mid \randF=f]$. To see this, consider any labeling $g'$ where $g'(v)=g(v)$ for $v \notin I^*_{fg}$ and either $g'(v) = g(v)-1$ or $g'(v) = g(v)$ for $v \in I^*_{fg}$. There are $2^{|I^*_{fg}|}$ such $g'$. We will prove that $\Pr[\randG = g' \mid \randF = f] \geq \Pr[\randG = g \mid \randF = f]$. For a classification $g \in \supp(\randG)$ and assignment to the edges $f \in \supp(\randF)$, define $E_{fg}^+$ as the subset of edges for which $(g(u)-g(v) - f(u,v)) \bmod k = 0$ and let $E_{fg}^- = E \setminus E_{fg}^+$. Since the noises on the edges are independent, it follows that 


\begin{align*}
\Pr[\randF = f \mid \randG=g] &= \left(\frac{1}{k} + \delta\right)^{|E_{fg}^+|}\left(\frac{1}{k} - \frac{\delta}{k-1}\right)^{|E_{fg}^-|}
\end{align*}

Comparing $g'$ and $g$, we notice that all edges $(u,w)$ with $v \notin \{u,w\}$ contribute the same to $\Pr[\randF = f \mid \randG = g]$ and $\Pr[\randF = f \mid \randG = g']$. However, for $g'$, it holds that any edge where $g(v)-g(u) \bmod k = 1$ we now have $g'(v)-g'(u) \bmod k = g(v) - 1 - g(u) \bmod k = 0$. Hence $c^{fg'}_v(0) = c^{fg}_v(1) \geq c^{fg}_v(0)$. It follows that $\Pr[\randF = f \mid \randG=g'] \geq \Pr[\randF = f \mid \randG=g]$. Using Bayes' theorem, we get
\begin{eqnarray*}
\Pr[\randG = g \mid \randF=f] = \frac{\Pr[\randF = f \mid \randG=g] \Pr[\randG=g]}{\Pr[\randF=f]}.
\end{eqnarray*}
and
\begin{eqnarray*}
\Pr[\randG = g' \mid \randF=f] = \frac{\Pr[\randF = f \mid \randG=g'] \Pr[\randG=g']}{\Pr[\randF=f]}.
\end{eqnarray*}
Since $\randG$ is uniform over its support, we have $\Pr[\randG=g] = \Pr[\randG=g']$. Hence we conclude that 
\begin{eqnarray*}
\Pr[\randG = g' \mid \randF = f] \geq \Pr[\randG = g \mid \randF = f]
\end{eqnarray*}
as claimed.

The above implies that $\Alg^*$ outputs $g$ with probability at most $2^{-|I^*_{fg}|}$ when it sees the query answers $f$. Indeed, if there is even a single $g'$ with $\Pr[\randG = g' \mid \randF = f] > \Pr[\randG = g \mid \randF = f]$, then $\Alg^*$ never outputs $g$, and otherwise, $\Alg^*$ outputs a uniform random labeling among the $2^{|I^*_{fg}|}$ candidates. To upper bound the succes probability of $\Alg^*$, we thus argue that $I^*_{\randF \randG}$ is large with high probability when $t$ is small.

Assume first that $1/n^{1/4} \leq \delta \leq 1/2k$ and $k \leq n^{o(1)}$. Using Lemma~\ref{lem:smalldelta}, each $v \in I$ is included in $I^*_{\randF \randG}$ with probability at least $c_2^{-1} \exp(-\delta^2 |E_v| k c_2)$ for a constant $c_2>0$. Since $|E_v| \leq 4t/n$, it follows that for $t = o((n \lg n)/(k \delta^2))$, $v$ will appear in $I^*_{\randF \randG}$ with probability at least $n^{-o(1)}$. Furthermore, $|I| \geq (n-1)/(4t/n + 1)  = \Omega(\delta^2nk / \lg n) = \Omega(n^{1/3})$. Moreover, these events are independent for different $v \in I$ since $I$ forms an independent set. A Chernoff bound implies that $|I^*_{\randF \randG}| \geq c_2^{-1} \exp(-\delta^2 |E_v| k c_2) |I|/2 \geq n^{1/3-o(1)}$ with probability at least $1-\exp(-n^{1/3-o(1)})$. When this event $B$ happens, the conditional success probability is no more than $\exp(-n^{1/3-o(1)})$. Hence the overall success probability is at most $\exp(-n^{1/3-o(1)})\Pr[B] + (1-\Pr[B]) = \exp(-n^{\Omega(1)})$.

Assume next that $1/2k \leq \delta \leq 1/4$ and $k \leq n^{o(1)}$. Using Lemma~\ref{lem:bigdelta}, each $v \in I$ occurs in $I^*_{\randF \randG}$ with probability at least $c_2^{-1} \exp(-\delta |E_v| c_2)$ for a constant $c_2 > 0$. Since $|E_v| \leq 4t/n$, it follows that for $t = o((n \lg n)/\delta)$, $v$ will appear in $I^*_{\randF \randG}$ with probability at least $n^{-o(1)}$. We also have $|I| \geq (n-1)/(4t/n + 1)  = \Omega(\delta n/ \lg n) = \Omega(n^{1/3})$. A Chernoff bound like above concludes that the success probability is no more than $\exp(-n^{\Omega(1)})$.

\subsection{Concentration Inequalities}
\label{sec:concentration}
In this section, we prove the two concentration inequalities stated in Section~\ref{sec:prelim}. Our proofs use the standard Chernoff bounds as well as the following ``reverse'' Chernoff bound:
\begin{theorem}[\cite{reversechernoff}]
Let $C_1,\dots,C_m$ be i.i.d. $0/1$ random variables with $\Pr[C_i = 1]=p$. For $p \leq 1/2$ and for any $0 \leq t \leq m(1-2p)$ it holds that:
$$
\Pr[\sum_{i=1}^m C_i \geq t + pm] \geq \frac{1}{4}\exp(-2t^2/pm).
$$
\end{theorem}
We start by proving Lemma~\ref{lem:smalldelta}:
\begin{proof}[Proof of Lemma~\ref{lem:smalldelta}]
For $0 \leq \delta \leq 1/2k$ and $n \geq k/2$, we first upper bound $\Pr[\sum_i X_i \leq 0]$. Let $Y_i$ take the value $1$ if $X_i$ takes the value $1$ and $0$ otherwise. Let $Z_i$ take the value $1$ if $X_i=-1$ and $0$ otherwise. By a Chernoff bound with $(1-\eps)(1/k+\delta)=(1/k+\delta/2) \Rightarrow \eps = \delta/(2(1/k+\delta))$, we get $\Pr[\sum_i Y_i \leq (1/k + \delta/2)n] \leq \exp(-\eps^2(1/k+\delta)n/2)$. This is at most $\exp(-\delta^2n/(8(1/k+\delta))) \leq \exp(-\delta^2nk/8)$. Similarly, a Chernoff bound with $(1+\eps)(1/k-\delta/(k-1))=(1/k+\delta/2) \Rightarrow \eps = (\delta/2 + \delta/(k-1))/(1/k - \delta/(k-1)) \geq \delta/(2(1/k-\delta/(k-1)))$, we have $\Pr[\sum_i Z_i \geq (1/k+\delta/2)n] \leq \exp(-\eps^2 (1/k-\delta/(k-1))n/3) \leq \exp(-\delta^2n/(12(1/k-\delta/(k-1)))) \leq \exp(-\delta^2 nk/12)$. A union bound gives $\Pr[\sum_i X_i \leq 0] \leq 2\exp(-\delta^2 nk/12)$. If $n \leq k/2$, then $\delta^2 n k < 1$ and the statement follows trivially since there is a constant $c_1$ making $c_1 \exp(-\delta^2 nk/c_1)$ greater than or equal to $1$.

To lower bound $\Pr[\sum_i X_i  \leq 0]$, let $W = \sum_i (Y_i + Z_i)$. Conditioned on $W = m$, we have that $\sum_i X_i$ is distributed as the sum of $m$ i.i.d. random variables taking the value $1$ with probability $(1/k+\delta)/(2/k + \delta k/(k-1)) \leq 1/2   + \delta k/2$ and the value $-1$ with probability at least $1/2-\delta k /2$. We will use the following ``reverse'' Chernoff bound:

Conditioned on $W=m > 0$, we ask what is the probability that we see at least $\lceil m/2\rceil$ $-1$'s, i.e. $\sum_i X_i \leq 0$. Fixing $t = \lceil m/2\rceil - (1/2 -\delta k/2 )m \leq \delta k m/2+ 1$ we see that
\hide{ 
\begin{eqnarray*}
\Pr[\sum_i X_i \leq 0 \mid W = m] &\geq& \\
\frac{1}{4}\exp(-2(\delta k m/2 + 1)^2/(1/2-\delta k /2)m) &\geq& \\
\frac{1}{4}\exp(-8(\delta k m/2 + 1)^2/m) &\geq& \\
\frac{1}{4}\exp(-16(\delta^2 k^2 m/4 + 1)).
\end{eqnarray*}
}

\begin{align*}
\Pr[\sum_i X_i \leq 0 \mid W = m] &\geq \frac{1}{4}\exp(-2(\delta k m/2 + 1)^2/(1/2-\delta k /2)m)  \geq \\ 
\frac{1}{4}\exp(-8(\delta k m/2 + 1)^2/m) &\geq \frac{1}{4}\exp(-16(\delta^2 k^2 m/4 + 1)).
\end{align*}
Using that $\Mean{W} = (2/k + \delta k/(k-1))n$, Markov's inequality gives us that $W \leq (4/k + 2 \delta k/(k-1))n \leq 8n/k$ with probability at least $1/2$. We also have $\sum_i X_i = 0$ when $\sum_i W_i = 0$. Hence

\begin{align*}
\Pr[\sum_i X_i \leq 0] &\geq \frac{1}{8}\exp(-32(\delta^2 kn + 1))  \geq  \frac{1}{8 \cdot e^{-32}}\exp(-32\delta^2 nk). 
\end{align*}

\hide{
\begin{eqnarray*}
\Pr[\sum_i X_i \leq 0] &\geq& \\
\frac{1}{8}\exp(-32(\delta^2 kn + 1)) &\geq& \\
\frac{1}{8 \cdot e^{-32}}\exp(-32\delta^2 nk). 
\end{eqnarray*}
}
\end{proof}

\noindent Next we prove Lemma~\ref{lem:bigdelta}:

\begin{proof}[Proof of Lemma~\ref{lem:bigdelta}]
We start by upper bounding $\Pr[\sum_i X_i \leq 0]$. Let $Y_i$ take the value $1$ if $X_i$ takes the value $1$ and $0$ otherwise. Let $Z_i$ take the value $1$ if $X_i=-1$ and $0$ otherwise. A Chernoff bound gives 

$$\Pr[\sum_i Y_i \leq (1/k + \delta/2)n] \leq \exp(-\delta^2n/(8(1/k+\delta))) \leq \exp(-\delta n/8).$$

\noindent Similarly, we have 

\begin{align*} 
\Pr[\sum_i Z_i \geq (1/k+\delta/2)n] &\leq \exp(-\delta^2n/(12(1/k-\delta/(k-1)))) \leq  \exp(-\delta^2n/(12(1/k-1/3(k-1)))) \\ 
&\leq \exp(-\delta^2n/(12(1/k-2/3k))) = \exp( -\delta^2kn/4) \leq \exp(-\delta n/8)
\end{align*}

\noindent A union bound gives $\Pr[\sum_i X_i \leq 0] \leq 2\exp(-\delta n /8)$.  To lower bound $\Pr[\sum_i X_i \leq 0]$, first notice that 

\begin{align*} 
\Pr[\sum_i Y_i =0] &= (1-1/k-\delta)^n \geq (1-3\delta)^n = \exp(-n \sum_{j=1}^\infty (3\delta)^j/j) \geq  \\ 
& \geq \exp(-n(3 \delta) \sum_{j=0}^\infty (3/4)^j) = \exp(-12 \delta n).
\end{align*}

\noindent We conclude that $\Pr[\sum_i X_i \leq 0] \geq \exp(-12 \delta n)$.
\end{proof}
 
\section{Conclusion}
\label{sec:concl} 
In this work we provide an optimal algorithm both in terms of running time and query complexity for  the problem of  learning joint alignments with a faulty oracle. The algorithm is simple and performs well in practice compared to previous work. An interesting open problem is to explore whether there exists an adaptive algorithm with better query complexity.  Finally, a remaining open question from 
Chen and Cand\'{e}s  is whether we can characterize the performance of existing joint alignment algorithms if one is satisfied with approximate solutions.

\newcommand{\etalchar}[1]{$^{#1}$}

\end{document}